\documentclass[11pt,letterpaper]{article}
\usepackage[T1]{fontenc}
\usepackage[utf8]{inputenc}
\usepackage{lmodern}
\usepackage[margin=1in]{geometry}
\usepackage{amsmath,amssymb,amsthm,mathtools}
\usepackage{microtype}
\usepackage{enumitem}
\usepackage{booktabs,array}
\usepackage{xcolor}
\usepackage[colorlinks=true,linkcolor=blue!45!black,citecolor=green!35!black,
  urlcolor=blue!55!black,bookmarksnumbered=true]{hyperref}
\usepackage[capitalize,noabbrev,nameinlink]{cleveref}
\usepackage{booktabs}
\usepackage{tabularx}

\hypersetup{pdftitle={Fast Odd-Permutation Sums in Characteristic Two and Shortest Even Directed Cycles},
  pdfauthor={Hanqing Li}}
\newtheorem{theorem}{Theorem}[section]
\newtheorem{lemma}[theorem]{Lemma}
\newtheorem{proposition}[theorem]{Proposition}

\theoremstyle{definition}

\theoremstyle{remark}

\numberwithin{equation}{section}
\DeclareMathOperator{\rank}{rank}

\DeclareMathOperator{\diag}{diag}

\newcommand{\F}{\mathbb F}
\newcommand{\softO}{\widetilde O}

\newcommand{\one}{\mathbf 1}

\setlist{topsep=4pt,itemsep=2pt,parsep=0pt}
\allowdisplaybreaks[1]

\usepackage{geometry}
\usepackage{hyperref}
\newif\ifanonymous
\anonymousfalse

\title{Fast Odd-Permutation Sums in Characteristic Two\\
and Shortest Even Directed Cycles}
\date{}

\ifanonymous
  \author{}
\else
  \author{Boxuan Chen\thanks{Peking University. Email: \texttt{chenboxuan@stu.pku.edu.cn}.}
  \and Bin Cui\thanks{Peking University. Email: \texttt{bin.cui@pku.edu.cn}.}
  \and Hanqing Li\thanks{Peking University. Email: \texttt{2400012952@stu.pku.edu.cn}.}
  \and Xupeng Miao\thanks{Peking University. Email: \texttt{xupeng.miao@pku.edu.cn}.}
  \and Haoze Wang\thanks{Unaffiliated. Email: \texttt{w97871749@gmail.com}}
  }
\fi

\begin{document}
\maketitle
\begin{abstract}
For an $n\times n$ matrix over a field of characteristic two,
let $\Phi$ denote the sum of the odd permutation monomials.
We give a deterministic algorithm that evaluates $\Phi$
and all its first partial derivatives using $O(n^\omega)$
field operations, where $2<\omega\le3$ is any fixed admissible
exponent for square matrix multiplication.

As an application, we obtain faster randomized algorithms
for shortest even directed cycles.
Given an $n$-vertex directed graph whose shortest even cycle
has length $\ell$, with high probability we compute $\ell$
and the union of the arcs of all shortest even cycles in
$\widetilde O(n^\omega\ell)$ time, and find one such cycle in
$\widetilde O(n^\omega\ell+n^2\ell^2)$ time.
The resulting general bounds are
$\widetilde O(n^{\omega+1})$ and $\widetilde O(n^4)$,
respectively, improving on the
$\widetilde O(n^{\omega+3})$ and
$\widetilde O(n^{\omega+4})$ bounds of
Bj\"orklund, Husfeldt, and Kaski~\cite{BHK24}.
\end{abstract}
\clearpage
\section{Introduction}
\label{sec:intro}

Over a field of characteristic two, the determinant
equals the permanent, and therefore computes the sum
of all permutation monomials. 
However, this identity does not separate the contributions of even and odd
permutations, since their signs coincide in this characteristic.
We study the odd-permutation sum
\begin{equation}
  \Phi(A)=\sum_{\substack{\pi\in S_n\\\pi\text{ odd}}}
              \prod_{i=1}^n A_{i,\pi(i)}.
  \label{eq:phi}
\end{equation}

The same polynomial appears as the parity cycle-cover enumerator
in the work of Bj\"orklund, Husfeldt, and Kaski~\cite{BHK24}.
Over sufficiently large binary extension fields $\mathbb F_{2^d}$,
they gave a randomized algorithm that evaluates $\Phi$ in
$\widetilde O(n^{\omega+2}d)$ time. Their approach lifts the
computation to a finite commutative ring of characteristic four,
where efficient permanent and determinant computations allow
the value of $\Phi$ to be recovered in characteristic two.

We give a deterministic algorithm that computes both $\Phi$
and all its first partial derivatives in $O(n^\omega)$ field
operations. The algorithm works over every field of
characteristic two, including $\mathbb F_2$, and applies to
matrices of arbitrary rank. Our bound is stated in field operations. Over $\mathbb{F}_{2^d}$,
with a standard field representation supplied, it yields
$\widetilde{O}(n^\omega d)$ bit operations.
Our algorithm not only improves the
previous evaluation bound by a factor of $n^2$ up to
polylogarithmic factors, but also computes the full gradient.

\begin{theorem}
\label{thm:algebra}
For every characteristic-two field $\F$ and every $A\in\F^{n\times n}$,
one can compute $\Phi(A)$ and
\[
  G(A)_{j,i}=\frac{\partial\Phi}{\partial A_{i,j}}(A)
\]
deterministically in $O(n^\omega)$ field operations, where $2<\omega\le 3$ is any fixed admissible
exponent for square matrix multiplication.
\end{theorem}

As an application, we obtain faster algorithms for shortest even
directed cycles. 
A permutation is odd precisely when its cycle decomposition
contains an odd number of even cycles. Thus, if $W$ is a
weighted adjacency matrix, the
lowest-degree terms of $\Phi(I+zW)$ encode shortest even directed
cycles whenever such cycles exist. 
Based on this observation, Bj\"orklund, Husfeldt, and Kaski~\cite{BHK24}
gave a randomized algorithm that computes the shortest even-cycle
length in $\widetilde{O}(n^{\omega+3})$ time and finds such a cycle
in $\widetilde{O}(n^{\omega+4})$ time. We improve both bounds and
additionally recover the union of the arcs of all shortest even
cycles.


Let $\ell$ be the length of a shortest even directed cycle.
Using our evaluator, we compute this length and the union of
the arcs of all shortest even cycles in
$\widetilde O(n^\omega\ell)$ time. We further find an explicit
shortest even cycle in
$\widetilde O(n^\omega\ell+n^2\ell^2)$ time.
Both guarantees hold with high probability.
In particular, the general bounds become
$\widetilde O(n^{\omega+1})$ for the length and
$\widetilde O(n^4)$ for a witness, with stronger bounds when
the shortest even cycle is short.



\begin{theorem}
\label{thm:support}
Given an $n$-vertex directed graph whose shortest even
cycle has length $\ell$, the length of the shortest even
cycle and the union of the arcs of all shortest even cycles can be found in
$\softO(n^\omega\ell)$ time, with probability $1-O(n^{-3})$.
\end{theorem}

\begin{theorem}
\label{thm:witness}
Given an $n$-vertex directed graph whose shortest even
cycle has length $\ell$, a shortest even directed cycle can be found in time $\softO(n^\omega\ell+n^2\ell^2)$, with probability $1-O(n^{-3})$.
\end{theorem}
The length-sensitive running-time bounds hold with the stated success probability. Every execution satisfies the corresponding general bound.

\subsection{Technical Overview}

\paragraph{Evaluation of invertible matrices.}
Permutation parity equals inversion parity. In characteristic two,
weighting each permutation by its inversion count expresses $\Phi$
as a sum of complementary minors. For invertible $A$, Jacobi's identity
rewrites these minors using $B=A^{-1}$ and $\det A$.
The resulting expression has two parts.
One resulting sum admits a quadratic-time prefix-sum computation.
The other is the crossed sum
\[
    T=
    \sum_{\substack{i<j\\k>l}}
    A_{i,k}B_{k,j}A_{j,l}B_{l,i}.
\]

We construct binary interval trees on the row and column indices. 
With row children $I,J$ and column children $L,K$, each pair of splits contributes
\[
    \operatorname{tr}\bigl(A[I,K]B[K,J]A[J,L]B[L,I]\bigr).
\]
If the child intervals have sizes $s$ and $t$, organizing the products
around the smaller dimension costs
$O(\max\{s,t\}\min\{s,t\}^{\omega-1})$ operations.
Summing over all pairs of tree levels gives $O(n^\omega)$ operations.

\paragraph{Derivatives and singular inputs.}
Fix one row of $A$.
Replacing the row with a vector $q$ makes $\Phi$ linear in $q$.
For invertible $A$, since its rows form a basis, we can recover the coefficients by evaluating this
linear function on the rows of $A$. More precisely, replacing a row by itself changes nothing; replacing a row by another creates
equal rows, and pairing permutations on these rows yields an explicit
formula. 
Given $A^{-1}$, $\det A$, and $\Phi(A)$, it computes one column of the gradient in $O(n^2)$ operations.
Further, the full gradient can be computed in $O(n^\omega)$ operations.

The value and gradient of $\Phi$ vanish at coranks at least three and four,
respectively. The remaining singular inputs reduce to a constant number of
invertible matrices by adding at most three rows and columns.
Boolean summation over the added diagonal entries recovers $\Phi$.
With the border fixed, this is a polynomial identity in the original
matrix entries, so the gradient can also be recovered.

\paragraph{Length and arc support.}
Assign an indeterminate $w_e$ to every arc $e$, and let $W(w)$ be
the weighted adjacency matrix. 
Consider $$Q(z,w)=\Phi(I+zW(w)).$$
Every contributing permutation contains an even cycle.
Consequently, the first nonzero coefficient occurs at the shortest
even-cycle length $\ell$ and is the sum of the distinct monomials
representing shortest even cycles.
After one random weight assignment over a sufficiently large field,
we evaluate $Q$ modulo $z^{L+1}$ at geometrically increasing precisions,
starting at $L = 1$ and doubling the precision until a coefficient of positive even degree is nonzero.
With high probability, doubling recovers $\ell$ in
$\widetilde O(n^\omega\ell)$ time.
Moreover, $w_e\cdot \partial Q/\partial w_e$ retains exactly the monomials
containing $e$, so one full-gradient computation identifies all
shortest-cycle arcs within the same bound.
\paragraph{Cycle extraction.}
Given the correct length $\ell$ and a nonzero initial arc
coefficient, we initialize $A=I+zW(\beta)$ using the weight
assignment from the length computation.
We first restrict one row to force the initial arc, and
then grow a directed path $P$ while maintaining
\[
    [z^\ell]\Phi(A)\ne0.
\]
The row restrictions ensure that this coefficient is the
weighted sum of shortest even cycles containing $P$.

At the current endpoint, we partition this sum according
to the next arc. The corresponding branch values can be
computed from one gradient column.
Since they sum to the maintained nonzero coefficient,
at least one branch is nonzero.
We select such an arc and restrict the endpoint row
accordingly, preserving the invariant.
Repeating this process closes a cycle of length $\ell$,
without introducing additional randomness.

To avoid recomputing $\Phi(A)$ after each restriction,
we use our gradient-column formulas and rank-one updates
on two fixed invertible bordered matrices.
These tools implement each extension in $O(n^2)$ ring
operations, or $\widetilde O(n^2\ell)$ time at precision
$\ell$. Including initialization, the total running time
is $\widetilde O(n^\omega\ell+n^2\ell^2)$.

\subsection{Related Work}

\paragraph{Cycle problems.}
Table~\ref{tab:cycle-algorithms} places our results alongside the
classical bounds summarized in~\cite[Table~1]{BHK24}.
The distinction between odd and even cycles is important:
a shortest odd closed walk is simple, whereas traversing an odd
cycle twice gives an even closed walk with repeated vertices.
Polynomial-time recognition of even-cycle-free digraphs was
established by McCuaig~\cite{McC04} and, independently,
Robertson, Seymour, and Thomas~\cite{RST99}.
Bj\"orklund, Husfeldt, and Kaski~\cite{BHK24} gave the first
polynomial-time algorithm for shortest even directed cycles,
using the parity cycle-cover viewpoint of Vazirani and
Yannakakis~\cite{VY89} and computations in characteristic four.
We retain this encoding but evaluate its polynomial directly
in characteristic two.

\begin{table}[t]
\centering
\small
\setlength{\tabcolsep}{4pt}
\renewcommand{\arraystretch}{1.12}
\begin{tabularx}{\linewidth}{@{}
    >{\raggedright\arraybackslash}X
    l l
    >{\raggedright\arraybackslash}p{0.23\linewidth}@{}}
\toprule
Problem & Graph & Time & Reference or method \\
\midrule
Shortest cycle length (girth)
    & Both & $O(nm)$ & Itai--Rodeh~\cite{IR78} \\
Odd-cycle detection
    & Both & $O(n+m)$ & Depth-first search \\
Shortest odd-cycle length
    & Both & $O(nm)$ & Repeated breadth-first search \\
\midrule
Even-cycle detection
    & Undirected & $O(n+m)$ & \cite{APY91} \\
    & Directed & $O(n^3)$ & \cite{RST99} \\
\midrule
Shortest even-cycle length
    & Undirected & $O(n^2)$ & \cite{YZ97} \\
    & Directed & $\widetilde O(n^{\omega+3})^{\dagger}$
    & \cite{BHK24} \\
    & Directed & $\widetilde O(n^{\omega+1})^{\dagger}$
    & \textbf{This work} \\
\midrule
Finding a shortest even cycle
    & Directed & $\widetilde O(n^{\omega+4})^{\dagger}$
    & \cite{BHK24} \\
    & Directed & $\widetilde O(n^4)^{\dagger}$
    & \textbf{This work} \\
\bottomrule
\end{tabularx}
\caption{Selected algorithms for cycle problems in simple unweighted
         graphs, following~\cite[Table~1]{BHK24}.
         A dagger marks a randomized algorithm.}
\label{tab:cycle-algorithms}
\end{table}

\paragraph{Algebraic methods and extraction.}
Our use of equal-row pairing is related to permanent computation
modulo powers of two, beginning with Valiant~\cite{Val79};
see also~\cite{BHL17}.
Baur and Strassen~\cite{BS83} established the general principle that
first derivatives can be computed with constant-factor arithmetic
overhead. Our explicit formulas additionally support $O(n^2)$
gradient-column queries once the inverse, determinant, and
$\Phi$-value are available, while fixed-border polynomial identities
handle singular inputs.
Substituting our general length algorithm into the self-reduction
of~\cite{BHK24} would give an
$\widetilde O(n^{\omega+2})$ extraction bound.
Our $\widetilde O(n^4)$ bound instead combines column queries with
rank-one updates, avoiding repeated calls to the full evaluator.

\subsection{Organization}
\Cref{sec:prelim} introduces the necessary algebraic preliminaries.
\Cref{sec:permutation-sum} presents an $O(n^\omega)$-time algorithm
for evaluating $\Phi$, and \cref{sec:gradient} computes the full gradient.
\Cref{sec:support} applies this algorithm to find the shortest
even-cycle length and the union of all shortest-cycle arcs.
Finally, \Cref{sec:witness} shows how to extract a shortest
even directed cycle efficiently.

\section{Preliminaries}
\label{sec:prelim}

All fields and rings in this paper have characteristic two, and rings are
commutative. $[n]:=\{1,\ldots,n\}$.
Submatrices $A[I,J]$ inherit the
original index orders, and
$A[-I,-J]=A[[n]\setminus I,[n]\setminus J]$.
$\det A[\emptyset,\emptyset] = 1$.
A matrix over a ring is invertible if and only if its determinant is a
unit. 

\paragraph{Parity.}
A permutation of $[n]$ with $b$ cycles, including fixed points, has sign
$(-1)^{n-b}$. Equivalently, it is odd precisely when its cycle decomposition
contains an odd number of even cycles. 

\paragraph{Multilinearity of $\Phi$.}
The polynomial $\Phi$ is multilinear in each row, in each column, and in
each entry. Its derivatives are formal polynomial derivatives. If
$A[i\leftarrow u]$ replaces row $i$ by a row vector $u$, then
\begin{equation}
  \Phi(A[i\leftarrow u])=u\,G(A)_{*,i}.
  \label{eq:row-linear}
\end{equation}
The coefficient column $G(A)_{*,i}$ is independent of row $i$ itself.

\paragraph{Linear algebra.}
Throughout, $2<\omega\le 3$ denotes any fixed exponent
such that two $n\times n$ matrices over a field can
be multiplied using $O(n^\omega)$ field operations.

For invertible $A$, put $B=A^{-1}$ and $\Delta=\det A$.
Jacobi's complementary-minor identity~\cite{HJ13} takes the form
\begin{equation}
  \det A[-\{i,j\},-\{k,l\}]
      =\Delta(B_{k,i}B_{l,j}+B_{k,j}B_{l,i})
  \quad(i\ne j,\ k\ne l).
  \label{eq:jacobi}
\end{equation}
All sign factors disappear in characteristic two. We use standard
$O(n^\omega)$ algorithms for the determinant, inverse, rank, and a nonsingular
submatrix whose order equals the rank of $A$~\cite{BH74,JPS13}. Schur complements, the matrix
determinant lemma, and the Sherman--Morrison update are used in their
standard forms~\cite{HJ13}. When applying them over rings we explicitly
check that all required denominators are units.


\paragraph{Finite fields and identity testing.}
For graph algorithms, $n\le1$ is handled directly. Otherwise we use
$\F=\F_{2^d}$ with $d=6\lceil\log_2 n\rceil$, and write
$q=|\F|\ge n^6$. Arithmetic takes $\softO(d)$ bit operations per field
operation~\cite{vzGG13}. An irreducible polynomial of degree $d$ over the
fixed base field can be constructed deterministically in time polynomial
in $d$~\cite{Sho90}. 

\begin{lemma}[Schwartz--Zippel Lemma~\cite{Sch80,Zip79}]
Let $P\in\F[x_1,\ldots,x_m]$ be a nonzero polynomial of
total degree at most $r$, where $\F$ is a finite field of
size $q$. If $a_1,\ldots,a_m$ are chosen independently and
uniformly at random from $\F$, then
\[
    \Pr\bigl[P(a_1,\ldots,a_m)=0\bigr]
    \le \frac{r}{q}.
\]
\end{lemma}



\paragraph{Truncated arithmetic.}
Write $R_L=\F[z]/(z^{L+1})$. An element of $R_L$ is a unit exactly when
its constant coefficient is nonzero, and a matrix over $R_L$ is invertible
exactly when its constant matrix is invertible. Multiplication and unit
inversion at precision $L+1$ take $\softO(L+1)$ field
operations~\cite{CK91,vzGG13}. These algorithms apply in characteristic
two. 
\section{Fast Evaluation of the Odd-Permutation Sum}
\label{sec:permutation-sum}
We first develop an \(O(n^\omega)\)-time algorithm for evaluating \(\Phi\) on invertible matrices. We then extend it to arbitrary matrices using a Boolean border construction. The resulting evaluation algorithm will also serve as the basis for computing the full gradient in the next section.
\subsection{Invertible Matrices}
\label{sec:invertible}

\begin{lemma}
\label{lem:inversion}
Every square matrix $A$ over a commutative ring of characteristic two
satisfies
\begin{equation}
  \Phi(A)=\sum_{i<j}\sum_{k>l}
     A_{i,k}A_{j,l}\det A[-\{i,j\},-\{k,l\}].
  \label{eq:inversion}
\end{equation}
Consequently, over a field,
\begin{equation}
  \rank A\le n-3\quad\Longrightarrow\quad\Phi(A)=0.
  \label{eq:scalar-vanishing}
\end{equation}
\end{lemma}

\begin{proof}
For a permutation $\pi$, the inversion count
$\operatorname{inv}(\pi)=|\{(i,j):i<j,\ \pi(i)>\pi(j)\}|$
has the same parity as $\pi$. Thus $\operatorname{inv}(\pi)$ in a characteristic-two ring is the indicator that $\pi$ is odd. Interchanging the resulting
finite sums gives
\begin{align*}
  \Phi(A)
   &=\sum_{\pi\in S_n}\sum_{i<j}
        \one_{\pi(i)>\pi(j)}\prod_h A_{h,\pi(h)}\\
&=
\sum_{i<j}
\sum_{\substack{\pi\in S_n\\\pi(i)>\pi(j)}}
\prod_h A_{h,\pi(h)}\\
   &=\sum_{i<j}\sum_{k>l}A_{i,k}A_{j,l}
        \sum_{\substack{\pi\in S_n\\\pi(i)=k,\ \pi(j)=l}}
                  \prod_{h\notin\{i,j\}}A_{h,\pi(h)}.
\end{align*}
The inner sum is the permanent of the complementary submatrix. It
equals its determinant in characteristic two, proving the identity.
At rank at most $n-3$, every minor of order $n-2$ vanishes.
\end{proof}

Now assume $A$ is invertible, and put $B=A^{-1}$ and $\Delta=\det A$.
Substituting \eqref{eq:jacobi} into \eqref{eq:inversion} yields
\begin{align}
  \Phi(A)&=\Delta(S+T),\label{eq:ST}\\
  S&=\sum_{i<j,\ k>l}A_{i,k}B_{k,i}A_{j,l}B_{l,j},\label{eq:S}\\
  T&=\sum_{i<j,\ k>l}A_{i,k}B_{k,j}A_{j,l}B_{l,i}.\label{eq:T}
\end{align}
Let
$X_{i,k}=A_{i,k}B_{k,i}$. Then 
\begin{equation}
S =  \sum_{\substack{i<j\\k>l}}
X_{i,k}X_{j,l}
=
\sum_j
\sum_l X_{j,l}
\sum_{k>l}
\sum_{i<j}X_{i,k}
.
\end{equation}
For all pairs $(l,j)\in [n]^2$, $\sum_{k>l}\sum_{i<j}X_{i,k}$ can be precomputed in $O(n^2)$ operations.
Thus the uncrossed sum $S$ takes $O(n^2)$ operations. 

It remains to evaluate the crossed sum $T$.
The main difficulty is the two ordering constraints
$i<j$ and $k>l$. We use two binary interval trees
to partition the ordered pairs satisfying these
constraints into disjoint rectangles. 

Let $N$ be the smallest power of two with $N\ge n$. Replacing $A$ by $A\oplus I_{N-n}$ preserves both $\Phi(A)$ and $\det A$, since every permutation contributing a nonzero monomial fixes all added indices. The padded matrix has inverse $B\oplus I_{N-n}$. Thus, after padding and relabeling, we may assume without loss of generality that $n=N=2^h$.

Build a complete binary interval tree on each of the row and column index sets $[n]$ as follows. The root represents the entire interval $[n]$. For each nonsingleton interval, create a left child representing its first half and a right child representing its second half, and recurse on both children. The process terminates at singleton intervals, which form the leaves. We identify each node with the interval it represents. Let $\mathcal R$ consist of the pairs $(I,J)$ of left and right children of internal nodes in the row tree. Define $\mathcal C$ analogously for the column tree, writing its pairs as $(L,K)$.

\begin{lemma}
\label{lem:block-traces}
\begin{equation}
  T=
  \sum_{(I,J)\in\mathcal R}
  \sum_{(L,K)\in\mathcal C}
  \operatorname{tr}\!\left(
    A[I,K]B[K,J]A[J,L]B[L,I]
  \right).
  \label{eq:block-traces}
\end{equation}
\end{lemma}

\begin{proof}
Every pair $i<j$ belongs to exactly one rectangle $I\times J$:
the intervals $I,J$ are the children of the lowest common ancestor
of leaves $i,j$ in the row tree. Similarly, every pair $l<k$ belongs
to exactly one rectangle $L\times K$ in the column tree. Moreover,
\[
  \operatorname{tr}\!\left(
    A[I,K]B[K,J]A[J,L]B[L,I]
  \right)
  =
  \sum_{\substack{i\in I,\ j\in J\\l\in L,\ k\in K}}
    A_{i,k}B_{k,j}A_{j,l}B_{l,i}.
\]
Thus every summand of $T$ occurs exactly once on the right-hand side
of \eqref{eq:block-traces}.
\end{proof}

\begin{lemma}
\label{lem:crossed-sum}
$T$ can be computed in $O(n^\omega)$
operations.
\end{lemma}

\begin{proof}
Fix $(I,J)\in\mathcal R$ and $(L,K)\in\mathcal C$, and write
$s=|I|=|J|$ and $t=|L|=|K|$. Then $s$ and $t$ are powers of two.
Suppose first that $s\le t$. Form the two $s\times s$ matrices
\[
  U=A[I,K]B[K,J],
  \qquad
  V=A[J,L]B[L,I].
\]
Partitioning the shared dimension into $t/s$ blocks of size $s$
reduces each product to $t/s$ square matrix multiplications.
Including the additions, computing $U$ and $V$ takes $O(ts^{\omega-1})$ operations.
The desired trace is then
\[
  \operatorname{tr}(UV)
  =\sum_{i\in I}\sum_{j\in J}U_{i,j}V_{j,i},
\]
which takes another $O(s^2)$ operations.
Thus $\operatorname{tr}(UV)$ takes $O(ts^{\omega-1})$ operations.

If $t<s$, cyclic invariance of the trace gives
\[
  \operatorname{tr}\!\left(
    A[I,K]B[K,J]A[J,L]B[L,I]
  \right)
  =
  \operatorname{tr}\!\left(
    B[K,J]A[J,L]B[L,I]A[I,K]
  \right).
\]
We instead form the two $t\times t$ products
$B[K,J]A[J,L]$ and $B[L,I]A[I,K]$.
The same argument gives a cost of $O(st^{\omega-1})$.
Hence the contribution of a single node pair can be evaluated in
\begin{equation}
  O\!\left(
    \max\{s,t\}\min\{s,t\}^{\omega-1}
  \right)
  \label{eq:block-trace-cost}
\end{equation}
operations.

Index tree depths from zero, so that the pair $([1,n/2],[n/2+1,n])$ has depth zero. At row depth $a$ and column depth $b$,
there are $2^{a+b}$ node pairs, with
\[
  s=\frac{n}{2^{a+1}},
  \qquad
  t=\frac{n}{2^{b+1}}.
\]
By \eqref{eq:block-trace-cost}, their total contribution takes
\[
  O\!\left(
    2^{a+b}\max\{s,t\}\min\{s,t\}^{\omega-1}
  \right)
  =
  O\!\left(
    n^\omega 2^{-(\omega-2)\max\{a,b\}}
  \right)
\]
operations. There are $2d+1$ depth pairs satisfying
$\max\{a,b\}=d$. Summing over all depths therefore gives
\[
  O\!\left(
    n^\omega
    \sum_{d=0}^{h-1}(2d+1)2^{-(\omega-2)d}
  \right)
  =O(n^\omega).
\]
\end{proof}

\begin{lemma}
\label{cor:invertible-evaluation}
Let $\mathbb F$ be a field of characteristic two.
For every invertible $A\in\mathbb F^{n\times n}$, the value $\Phi(A)$ can be computed in
$O(n^\omega)$ field operations.
\end{lemma}

\begin{proof}
Compute $B=A^{-1}$ and $\Delta=\det A$ in $O(n^\omega)$ field
operations. The sum $S$ takes $O(n^2)$ operations, and
Lemma~\ref{lem:crossed-sum} evaluates $T$ in $O(n^\omega)$ operations.
The result follows from \eqref{eq:ST}.
\end{proof}

\subsection{Singular Matrices via a Boolean Border}
\label{sec:border}

We now extend the evaluation algorithm to singular matrices.
By Lemma~\ref{lem:inversion}, $\Phi(A)=0$ whenever
$\rank A\le n-3$. Thus only coranks one and two require additional
treatment. We first give a border construction valid for arbitrary
corank.

Let $A\in\mathbb F^{n\times n}$ have rank $r=n-k$.
Choose $I,J\subseteq[n]$ with $|I|=|J|=r$ such that
$A[I,J]$ is invertible. Write the complementary indices in increasing
order as
\[
  [n]\setminus I=\{i_1,\ldots,i_k\},
  \qquad
  [n]\setminus J=\{j_1,\ldots,j_k\},
\]
and define
\[
  U=(e_{i_1},\ldots,e_{i_k}),
  \qquad
  V=(e_{j_1},\ldots,e_{j_k}),
\]
where $e_i$ denotes the $i$th standard basis vector.
For an $n\times n$ matrix $X$ and
$\boldsymbol t=(t_1,\ldots,t_k)$, let
\begin{equation}
  M_X(\boldsymbol t)=
  \begin{pmatrix}
    X & U\\
    V^T & \diag(t_1,\ldots,t_k)
  \end{pmatrix}.
  \label{eq:boolean-border}
\end{equation}
The matrices $U,V$ are chosen from $A$ and then held fixed as $X$
varies. The original block $X$ retains its original row and column
orders.

\begin{lemma}
\label{lem:boolean-border}
For $\boldsymbol t\in\mathbb F^k$,
the matrix $M_A(\boldsymbol t)$ is invertible.
Moreover, with $U,V$ fixed, the following identity holds:
\begin{equation}
  \Phi(X)=
  \sum_{\boldsymbol\varepsilon\in\{0,1\}^k}
    \Phi\!\left(M_X(\boldsymbol\varepsilon)\right).
  \label{eq:boolean-identity}
\end{equation}
\end{lemma}

\begin{proof}
We first prove invertibility.
Permute the first $n$ rows and columns so that the indices in $I$
and $J$, respectively, come first, preserving the inherited orders
within each group. The resulting matrix has the form
\[
  \begin{pmatrix}
    P & Q & 0\\
    R & S & I_k\\
    0 & I_k & D
  \end{pmatrix},
  \qquad
  P=A[I,J],
  \qquad
  D=\diag(t_1,\ldots,t_k).
\]
Since $P$ is invertible, eliminating the first block row and column leaves
the Schur complement
\[
  \begin{pmatrix}
    S+RP^{-1}Q & I_k\\
    I_k & D
  \end{pmatrix}.
\]
The equality $\rank A=\rank P=r$ implies
$S+RP^{-1}Q=0$. Thus the remaining matrix is
\[
  \begin{pmatrix}
    0 & I_k\\
    I_k & D
  \end{pmatrix},
\]
which is invertible for every $D$.
Therefore $M_A(\boldsymbol t)$ is invertible.

We next prove \eqref{eq:boolean-identity}.
Each variable $t_j$ occurs in a single diagonal entry, so
$\Phi(M_X(\boldsymbol t))$ is multilinear in
$t_1,\ldots,t_k$. Write
\[
  \Phi(M_X(\boldsymbol t))
  =
  \sum_{S\subseteq[k]}c_S(X)\prod_{j\in S}t_j.
\]
A permutation monomial containing $t_1\cdots t_k$ fixes every added
index. Its restriction to the first $n$ indices has the same parity,
and every permutation on these indices extends uniquely in this
way. Hence
\[
  c_{[k]}(X)
  =
  [t_1\cdots t_k]\Phi(M_X(\boldsymbol t))
  =
  \Phi(X).
\]
For every $S\subseteq[k]$,
\[
  \sum_{\boldsymbol\varepsilon\in\{0,1\}^k}
    \prod_{j\in S}\varepsilon_j
  =
  2^{k-|S|}.
\]
In characteristic two, this is zero unless $S=[k]$, in which case
it is one. Summing over the Boolean cube therefore leaves exactly
$c_{[k]}(X)=\Phi(X)$, proving the identity.
\end{proof}

The row and column permutations in the proof are used only to
establish invertibility. Evaluating the right-hand side of
\eqref{eq:boolean-identity} uses the matrices in
\eqref{eq:boolean-border} with their original ordering, since
independent row and column permutations need not preserve $\Phi$.
Also, although the border is selected using $A$, the polynomial
identity holds for every $X$ with that border fixed. This will allow
us to differentiate the identity in the next section.

\begin{proposition}
\label{prop:scalar}
Let $\mathbb F$ be a field of characteristic two.
For every $A\in\mathbb F^{n\times n}$, the value $\Phi(A)$ can be
computed in $O(n^\omega)$ field operations.
\end{proposition}

\begin{proof}
For $n\le1$, return zero.
Otherwise, compute the rank $r$ of $A$ and a nonsingular
$r\times r$ submatrix in $O(n^\omega)$ field operations.
Put $k=n-r$.

If $k=0$, apply Lemma~\ref{cor:invertible-evaluation}.
If $k\ge3$, return zero by Lemma~\ref{lem:inversion}.
It remains to consider $k\in\{1,2\}$.
Construct $U,V$ from the selected submatrix and use
Lemma~\ref{lem:boolean-border} to write
\[
  \Phi(A)=
  \sum_{\boldsymbol\varepsilon\in\{0,1\}^k}
    \Phi\!\left(M_A(\boldsymbol\varepsilon)\right).
\]
Every matrix in this sum is invertible and has order $n+k$.
There are at most $2^k\le4$ such matrices, so
Lemma~\ref{cor:invertible-evaluation} evaluates all their
$\Phi$-values in
\[
  O\!\left(2^k(n+k)^\omega\right)=O(n^\omega)
\]
field operations. Adding these values gives $\Phi(A)$.
\end{proof}
\section{An Explicit Gradient at Every Rank}
\label{sec:gradient}
We develop two algorithms for accessing the gradient
\[
    G(A)_{j,i}=\frac{\partial\Phi}{\partial A_{i,j}}(A).
\]
First, given an invertible matrix $A$ together with
$A^{-1}$, $\det A$, and $\Phi(A)$, we show how to
compute any specified column of $G(A)$ in $O(n^2)$
operations. This will be used for repeated gradient
queries during cycle extraction in \Cref{sec:witness}.
Second, we compute the entire gradient in
$O(n^\omega)$ operations for matrices of arbitrary
rank. This allows us to identify all arcs belonging
to shortest even directed cycles in \Cref{sec:support}.

We first derive the gradient formula for invertible
matrices, and then extend it to singular matrices
using the Boolean border from \Cref{sec:permutation-sum}.
We recover each gradient column from the values
obtained by replacing its associated row with other
rows of $A$.

\begin{lemma}
\label{lem:proportional}
Let $A$ be an invertible matrix over a commutative ring of
characteristic two, and put $B=A^{-1}$ and $\Delta=\det A$.
For distinct $p,q$, let $N$ be obtained from $A$ by replacing
row $q$ with $tA_{p,*}$. Then
\begin{equation}
  \Phi(N)
  =
  t\Delta\sum_j A_{p,j}^2B_{j,p}B_{j,q}.
  \label{eq:proportional}
\end{equation}
\end{lemma}

\begin{proof}
Fix an unordered pair of columns $\{j,k\}$ assigned to rows $\{p,q\}$,
and fix a bijection between the remaining rows and columns.
There are two ways to assign $\{j,k\}$ to $\{p,q\}$; the resulting permutations have opposite parity, while their
monomial weights are equal: both contain the factor $tA_{p,j}A_{p,k}$.
Exactly one of these permutations therefore contributes to $\Phi(N)$.

Summing over the remaining bijections gives the permanent of the
complementary submatrix, which equals its determinant in
characteristic two. Thus
\[
  \Phi(N)
  =
  t\sum_{j<k}A_{p,j}A_{p,k}
    \det A[-\{p,q\},-\{j,k\}].
\]
Applying \eqref{eq:jacobi} yields
\[
  \Phi(N)
  =
  t\Delta\sum_{j<k}A_{p,j}A_{p,k}
    (B_{j,p}B_{k,q}+B_{k,p}B_{j,q}).
\]
Moreover,
\begin{align*}
  \sum_{j<k}A_{p,j}A_{p,k}
    (B_{j,p}B_{k,q}+B_{k,p}B_{j,q})
  &=
  \sum_{j\ne k}A_{p,j}A_{p,k}B_{j,p}B_{k,q}\\
  &=
  \left(\sum_j A_{p,j}B_{j,p}\right)
  \left(\sum_k A_{p,k}B_{k,q}\right)
  +
  \sum_j A_{p,j}^2B_{j,p}B_{j,q}.
\end{align*}
Since $AB=I$ and $p\ne q$,
\[
  \sum_k A_{p,k}B_{k,q}=(AB)_{p,q}=0.
\]
The product of the two sums therefore vanishes, leaving
\[
  \Phi(N)
  =
  t\Delta\sum_j A_{p,j}^2B_{j,p}B_{j,q},
\]
as required.
\end{proof}

\subsection{Compute a Specified Column of $G(A)$}
For a row vector $a$, let $A[q\leftarrow a]$ denote the matrix
obtained by replacing row $q$ of $A$ with $a$.
Every permutation monomial contains exactly one entry from row $q$,
so $\Phi(A[q\leftarrow a])$ is linear in $a$. Its coefficient
of $a_j$ is $\partial\Phi/\partial A_{q,j}$, which depends only
on the other rows. Therefore
\begin{equation}
  \Phi(A[q\leftarrow a])
  =
  \sum_{j=1}^n a_jG(A)_{j,q}
  =
  aG(A)_{*,q}.
  \label{eq:row-replacement}
\end{equation}
In particular, this coefficient column $G(A)_{*,q}$ is independent of row $A_{q,*}$.

Suppose $A$ is invertible, and put $B=A^{-1}$ and $\Delta=\det A$.
Fix $q\in[n]$, and let $g=G(A)_{*,q}$. For $p\in[n]$, define the column vector $h$ by
\[
  h_p:=\Phi(A[q\leftarrow A_{p,*}]).
\]

By \Cref{eq:row-replacement}, $h=Ag$, and hence $g=Bh$.
%
Replacing row $q$ by itself gives $h_q=\Phi(A)$.
For $p\ne q$, by Lemma~\ref{lem:proportional},
\begin{equation}
  h_p=\Delta\sum_j A_{p,j}^2B_{j,p}B_{j,q}.
  \label{eq:row-query}
\end{equation}
Computing all entries of $h$ takes $O(n^2)$ operations, as does
multiplying $B$ by $h$. Thus, once $A,B,\Delta,\Phi(A)$ are available,
a specified column of $G(A)$ can be computed in $O(n^2)$ operations.



\subsection{Compute All Columns Simultaneously}
Define the matrix $C$ by
\[
  C_{p,j}=A_{p,j}^2B_{j,p}.
\]

\begin{proposition}
\label{prop:gradient}
For every invertible $A$ over a commutative ring of
characteristic two,
\begin{equation}
  G(A)=\Delta BCB+(\Phi(A)+\Delta)B.
  \label{eq:gradient}
\end{equation}
\end{proposition}

\begin{proof}
Let $H$ be the matrix whose $q$th column is the vector $h$ in
\eqref{eq:row-query}. For $p\ne q$,
\[
  H_{p,q}
  =
  \Delta\sum_j A_{p,j}^2B_{j,p}B_{j,q}
  =
  \Delta(CB)_{p,q},
\]
whereas $H_{q,q}=\Phi(A)$.
On the diagonal of $CB$, we have
\[
  (CB)_{p,p}
  =
  \sum_j A_{p,j}^2B_{j,p}^2
  =
  \left(\sum_j A_{p,j}B_{j,p}\right)^2
  =
  1.
\]
The second equality uses characteristic two, and the last uses
$AB=I$. Consequently,
\[
  H=\Delta CB+(\Phi(A)+\Delta)I.
\]
Applying $g=Bh$ to every column gives
\[
  G(A)=BH=\Delta BCB+(\Phi(A)+\Delta)B.
\]
\end{proof}

We now consider singular matrices.
Note that the vanishing of $\Phi$ at corank at least three does not imply
the same vanishing for its derivatives, since changing an entry may increase the rank by one.
Thus we need to consider the case with corank at most three.

\begin{proof}[Proof of Theorem~\ref{thm:algebra}]
For $n\le1$, both the value and gradient are zero.
Otherwise, compute the rank of $A$ and a nonsingular submatrix
whose order equals the rank of \(A\) in $O(n^\omega)$ field operations.
Write $k=n-\rank A$.

If $k=0$, compute $B=A^{-1}$, $\Delta=\det A$, and $F=\Phi(A)$
in $O(n^\omega)$ field operations.
Proposition~\ref{prop:gradient} then computes $G(A)$ in another
$O(n^\omega)$ field operations.

Suppose next that $k\ge4$.
Replacing one row increases the rank by at most one, so for every
row vector $a$ and every row index $q$,
\[
  \rank(A[q\leftarrow a])
  \le \rank A+1
  \le n-3.
\]
Lemma~\ref{lem:inversion} therefore gives
$\Phi(A[q\leftarrow a])=0$.
Taking $a=e_j^T$ in \eqref{eq:row-replacement} yields
\[
  G(A)_{j,q}
  =
  \Phi(A[q\leftarrow e_j^T])
  =
  0.
\]
Thus the entire gradient is zero, as is $\Phi(A)$.

It remains to consider $1\le k\le3$.
Choose $U,V$ from the selected nonsingular submatrix as in
\eqref{eq:boolean-border}.
With this border fixed, Lemma~\ref{lem:boolean-border} gives
the polynomial identity
\[
  \Phi(X)
  =
  \sum_{\boldsymbol\varepsilon\in\{0,1\}^k}
    \Phi(M_X(\boldsymbol\varepsilon)).
\]
Differentiate with respect to each entry of $X$ and then set $X=A$.
Each entry $X_{i,j}$ occurs only in the corresponding position
of the upper-left block of $M_X(\boldsymbol\varepsilon)$.
Hence
\begin{equation}
  G(A)
  =
  \sum_{\boldsymbol\varepsilon\in\{0,1\}^k}
    G(M_A(\boldsymbol\varepsilon))[[n],[n]].
  \label{eq:gradient-border}
\end{equation}
The border is held fixed during differentiation. Since the identity
holds for every $X$, this argument remains valid at the singular
matrix $A$, without differentiating the rank computation or the
choice of $U,V$.

Every bordered matrix is invertible by
Lemma~\ref{lem:boolean-border}.
There are at most $2^k\le8$ such matrices, each of order at most
$n+3$. Computing their values and gradients, and summing the
required entries, takes $O(n^\omega)$ field operations.
The scalar values recover $\Phi(A)$ through
\eqref{eq:boolean-identity}, and the upper-left gradient blocks
recover $G(A)$ through \eqref{eq:gradient-border}.
\end{proof}

\section{Shortest Even Cycles and Their Arc Support}
\label{sec:support}

Let $D=([n],E)$ be a simple directed graph with $m$ arcs.
We assume that $D$ has no loops, but both $(u,v)$ and $(v,u)$
may be present. All directed cycles are simple and are identified
up to cyclic rotation. The case $n\le 1$ is immediate, so assume
that $n\ge 2$.

We work over the field
\[
  \mathbb F=\mathbb F_{2^d},
  \qquad
  d=6\lceil\log_2 n\rceil,
\]
and write $q=|\mathbb F|$. Thus $q\ge n^6$, and a field element
is represented using $O(\log n)$ bits. The cost of field arithmetic
is therefore absorbed by the polylogarithmic factors in our
$\widetilde O$ bounds.

Associate an independent indeterminate $w_e$ with each arc $e\in E$,
and define the weighted adjacency matrix by
\[
  W(\mathbf w)_{u,v}
  =
  \begin{cases}
    w_{(u,v)}, & (u,v)\in E,\\
    0, & (u,v)\notin E.
  \end{cases}
\]
Consider the polynomial
\begin{equation}
  Q(z,\mathbf w)=\Phi(I+zW(\mathbf w)).
  \label{eq:forward-polynomial}
\end{equation}
A permutation monomial that moves exactly $r$ vertices contributes
a term of degree $r$ in $z$. Its fixed points use the diagonal
entries of $I$, while its nontrivial cycles use arcs of $D$.
In particular, $\deg_z Q\le n$.

\begin{lemma}
\label{lem:extremal}
If $D$ has no even directed cycle, then $Q(z,\mathbf w)=0$.
Otherwise, let $\ell$ be the length of a shortest even directed
cycle, and let $\mathcal C_\ell$ be the set of all such cycles.
Then
\begin{equation}
  [z^r]Q(z,\mathbf w)=0
  \quad (r<\ell),
  \qquad
  [z^\ell]Q(z,\mathbf w)
  =
  \sum_{C\in\mathcal C_\ell}\prod_{e\in E(C)}w_e
  \ne 0.
  \label{eq:extremal}
\end{equation}
\end{lemma}

\begin{proof}
A permutation is odd if and only if its cycle decomposition
contains an odd number of even cycles. Thus every nonzero
permutation monomial contributing to $Q$ contains at least one
even directed cycle of $D$. If $D$ has no such cycle, there are
no contributing monomials.

Now suppose that the shortest even directed cycle has length
$\ell$. Every contributing permutation moves at least $\ell$
vertices, which proves the vanishing of all coefficients of
degree less than $\ell$. If exactly $\ell$ vertices are moved,
the permutation must consist of one directed cycle of length
$\ell$ and fixed points on all remaining vertices. Conversely,
each cycle in $\mathcal C_\ell$ defines such an odd permutation.
This proves the coefficient identity.

Distinct directed cycles have distinct arc sets and hence
contribute distinct squarefree monomials. Since
$\mathcal C_\ell$ is nonempty, their sum is a nonzero polynomial.
\end{proof}

For an existing arc $e=(u,v)\in E$, define
\begin{equation}
Q_e(z,w)
:=w_e\frac{\partial Q}{\partial w_e}(z,w)
=zw_eG(I+zW(w))_{v,u}.
  \label{eq:Qe}
\end{equation}
These polynomials record which arcs occur in the contributing
permutations.

\begin{lemma}
\label{lem:arc-support}
Suppose that the shortest even directed cycle has length $\ell$.
For every $e\in E$,
\begin{equation}
  [z^\ell]Q_e(z,\mathbf w)
  =
  \sum_{\substack{C\in\mathcal C_\ell\\ e\in E(C)}}
       \prod_{f\in E(C)}w_f.
  \label{eq:arc-coefficient}
\end{equation}
Consequently, this coefficient is a nonzero polynomial if and
only if $e$ belongs to a shortest even directed cycle.
\end{lemma}

\begin{proof}
By the definition of the transposed gradient and the chain rule,
\[
  \frac{\partial Q}{\partial w_e}
  =
  zG(I+zW(\mathbf w))_{v,u}.
\]
Hence
\[
  Q_e(z,\mathbf w)
  =
  w_e\frac{\partial Q}{\partial w_e}.
\]
The polynomial $Q$ is multilinear in the arc variables.
Multiplication of its derivative by $w_e$ therefore retains
exactly those monomials that contain $w_e$.

Applying this operation to the coefficient in
\eqref{eq:extremal} gives \eqref{eq:arc-coefficient}.
If no shortest even cycle contains $e$, the sum is empty.
Otherwise, its summands are distinct squarefree monomials,
so the sum is nonzero.
\end{proof}

We compute the required coefficients by evaluating $\Phi$ and
its gradient over a truncated polynomial ring. For an integer
$L\ge 0$, put
\[
  \mathcal R_L=\mathbb F[z]/(z^{L+1}).
\]
An element of $\mathcal R_L$ is represented by its coefficients
of degrees $0,\ldots,L$. Multiplication in $\mathcal R_L$ and
inversion of an element with nonzero constant coefficient take
$\widetilde O(L+1)$ field operations using standard polynomial
arithmetic~\cite{CK91,vzGG13}.

\begin{lemma}
\label{lem:series}
Let $A\in\mathcal R_L^{s\times s}$, and suppose that $A(0)$ is
invertible over $\mathbb F$. Then
$A^{-1},\ \det A,\ \Phi(A),\ G(A)$
can be computed in
$\widetilde O(s^\omega(L+1))$ field operations.
\end{lemma}

\begin{proof}
A bilinear matrix-multiplication algorithm over $\mathbb F$
also applies over the commutative $\mathbb F$-algebra $R_L$.
Thus multiplying two $s\times s$ matrices over $R_L$ takes
$O(s^\omega)$ ring operations.

We first compute $A^{-1}$ and $\det A$. Compute $A(0)^{-1}$
and $\det A(0)$ over $\mathbb F$, and set
\[
    \widehat A=A(0)^{-1}A.
\]
Then $\widehat A(0)=I$. We compute its inverse and determinant
by recursive block elimination. At each nonscalar step,
partition the current matrix as
\[
    \begin{pmatrix}
        P&Q\\
        S&T
    \end{pmatrix},
\]
with the two diagonal blocks of nearly equal size.
Since its constant matrix is the identity,
\[
    P(0)=T(0)=I,
    \qquad
    Q(0)=S(0)=0.
\]
Hence $P$ is invertible, and the Schur complement
\[
    H=T+SP^{-1}Q
\]
also satisfies $H(0)=I$. We recursively compute the inverses
and determinants of $P$ and $H$, and combine them using the
standard block formulas. Both recursive inputs therefore
have the identity matrix as their constant matrix. At each scalar leaf,
the element being inverted has constant coefficient one
and is a unit.

Each recursive step uses a constant number of matrix
multiplications at the current size. Consequently, the
recursion takes $O(s^\omega)$ ring additions and
multiplications, together with $O(s)$ inversions of scalar
units. We then recover
\[
    A^{-1}=\widehat A^{-1}A(0)^{-1},
    \qquad
    \det A=\det A(0)\det\widehat A.
\]

For the original matrix $A$, the complementary-minor
expansion, Jacobi's identity, and the block-trace
decomposition remain valid over $R_L$, since $A$ is
invertible. They evaluate $\Phi(A)$ using $O(s^\omega)$
ring operations. The explicit gradient formula likewise
holds over $R_L$ and computes $G(A)$ using two matrix
multiplications and $O(s^2)$ additional ring operations.

Each ring addition, multiplication, or inversion of a unit
takes $\widetilde O(L+1)$ field operations. Including the
initial computations over $\mathbb F$, all four quantities
can therefore be computed in
$\widetilde O(s^\omega(L+1))$ field operations.
\end{proof}

We now use increasing truncation precisions to find the first
nonzero even coefficient of $Q$. The same assignment of arc
weights is used throughout the computation.

\begin{proof}[Proof of Theorem~\ref{thm:support}]
Choose each \(\beta_e\), \(e\in E\), independently and uniformly at random from \(F\).
Write $W_\beta= W(\beta)$ and
\[
  Q_\beta(z)=\Phi(I+zW_\beta).
\]
For each precision
\[
  L=1,2,4,8,\ldots,
\]
with the last value capped at $n$, apply Lemma~\ref{lem:series} to compute
\[
  F_L=\Phi(I+zW_\beta)
  \text{ in }\mathcal R_L \text{, where }\mathcal R_L=\mathbb F[z]/(z^{L+1}),
\]
which takes $\widetilde O(n^\omega(L+1))$ field operations.

Inspect the coefficients of positive even degrees at most $L$.
If one is nonzero, let $r$ be the smallest such degree and stop
increasing the precision. If no trial succeeds by $L\ge n$, report
that $D$ has no even directed cycle.

Suppose that a trial succeeds and returns $r$.
Apply Lemma~\ref{lem:series} at precision $r$ to compute
\[
  \mathcal G_r
  =
  G(I+zW_\beta)
  \quad\text{over }\mathcal R_r.
\]
For each original arc $e=(u,v)\in E$, extract
\begin{equation}
  c_e
  =
  \beta_e[z^{r-1}](\mathcal G_r)_{v,u}
  =
  [z^r]\bigl(z\beta_e(\mathcal G_r)_{v,u}\bigr).
  \label{eq:forward-support-extraction}
\end{equation}
Return $r$ together with the set
\[
  E'=\{e\in E:c_e\ne 0\}.
\]
Computing $\mathcal G_r$ takes
$\widetilde O(n^\omega(r+1))$ field operations, and extracting
all $c_e$ takes another $O(m)$ field operations.

To prove correctness, first suppose that $D$ contains an even
directed cycle, and let $\ell$ be the shortest such length.
Define
\[
  P_\ell(\mathbf w)
  =
  [z^\ell]Q(z,\mathbf w),
  \qquad
  P_{\ell,e}(\mathbf w)
  =
  [z^\ell]Q_e(z,\mathbf w),
\]
and let
\[
  E_\ell
  =
  \bigcup_{C\in\mathcal C_\ell}E(C).
\]
By Lemma~\ref{lem:extremal}, $P_\ell$ is a nonzero polynomial
of total degree $\ell$. By Lemma~\ref{lem:arc-support},
$P_{\ell,e}$ is nonzero of total degree $\ell$ for every
$e\in E_\ell$, and is identically zero for every
$e\notin E_\ell$.

Consider the event
\[
  \mathcal E
  =
  \{P_\ell(\beta)\ne 0\}
  \cap
  \bigcap_{e\in E_\ell}
       \{P_{\ell,e}(\beta)\ne 0\}.
\]
The Schwartz--Zippel bound and a union bound give
\begin{equation}
  \Pr[\mathcal E^c]
  \le
  \frac{(1+|E_\ell|)\ell}{q}
  \le
  \frac{(m+1)\ell}{q}
  =
  O(n^{-3}).
  \label{eq:support-error}
\end{equation}
Here we used $m\le n(n-1)$, $\ell\le n$, and $q\ge n^6$.

Assume that $\mathcal E$ occurs.
All coefficients of $Q_\beta$ below degree $\ell$ vanish.
The algorithm therefore stops at the first precision $L\ge\ell$ and returns $r=\ell$.

At this value of $r$, equation
\eqref{eq:forward-support-extraction} gives
\[
  c_e=P_{\ell,e}(\beta).
\]
Every arc in $E_\ell$ has a nonzero extracted coefficient on
$\mathcal E$, while every arc outside $E_\ell$ has coefficient
zero for every assignment of the weights. Thus $E'=E_\ell$.

It remains to bound the running time.
On $\mathcal E$, the final precision satisfies
\[
  \ell\le L<2\ell.
\]
The precisions increase geometrically, so their sum is
$O(\ell)$. The total cost of finding the length is therefore
\[
  \sum_{\text{trials }L}
       \widetilde O(n^\omega(L+1))
  =
  \widetilde O(n^\omega\ell).
\]
The subsequent gradient computation has the same bound.
Consequently, with probability $1-O(n^{-3})$, the algorithm
returns the shortest even-cycle length and all arcs belonging
to shortest even cycles in
$\widetilde O(n^\omega\ell)$ time.

Regardless of the weight assignment, the truncation
precision never exceeds $n$, and the sum of all tested
precisions is $O(n)$. The subsequent gradient computation,
if performed, uses precision $r\le n$ and also takes
$\widetilde O(n^{\omega+1})$ time. Thus every execution
takes $\widetilde O(n^{\omega+1})$ time.

Finally, if $D$ has no even directed cycle, then
$Q(z,\mathbf w)=0$ by Lemma~\ref{lem:extremal}.
Every trial consequently returns the zero polynomial, and the
algorithm correctly reports absence after reaching precision
$n$, in $\widetilde O(n^{\omega+1})$ time.
\end{proof}
\section{Extracting a Shortest Even Directed Cycle}
\label{sec:witness}

We now extract a shortest even directed cycle using the
length and arc coefficients computed in
Section~\ref{sec:support}.
The algorithm starts from an arc with a nonzero coefficient
and grows a directed path by selecting nonzero outgoing-arc
contributions. When the returned length is correct, each
such contribution certifies that the entire selected path
is contained in a shortest even cycle.

We first describe the algorithm and establish structural
properties that hold for every weight assignment.
We then prove that, given the correct length and a nonzero
initial arc coefficient, the algorithm returns a shortest
even cycle. Finally, we show how to implement its queries
and updates efficiently. The extraction uses the same
weight assignment throughout and introduces no additional
randomness.

\subsection{The Extraction Algorithm}

Run the algorithm of Section~\ref{sec:support}, retaining its assignment
$\beta$ and its extracted arc coefficients.
If it reports absence, return the same answer.
Otherwise, let $r$ be the returned length, and work over
\[
  \mathcal R_r=\mathbb F[z]/(z^{r+1}).
\]
Write
\[
  A=I+zW_{\boldsymbol\beta}.
\]
For every original arc $e=(u,v)$, the computation in
Section~\ref{sec:support} provides
\[
  \alpha_e
  =
  [z^r]\bigl(z\beta_eG(A)_{v,u}\bigr).
\]
Fix arbitrary orderings of the vertices and arcs, which will be used
to resolve all choices below.
Whenever a required nonzero coefficient does not exist,
or the final verification fails, the algorithm reports
failure.

Choose the first arc $e_0=(u,v)$ with $\alpha_{e_0}\ne 0$.
Initialize the path and its coefficient by
\[
  P=(u,v),
  \qquad
  c=\alpha_{e_0}.
\]
To force the initial arc, replace row $u$ of $A$ by a row containing
only the entry corresponding to $u\to v$:
\[
  \widehat A
  =
  A[u\leftarrow z\beta_{e_0}e_v^\top].
\]
In particular, this restriction removes the diagonal entry in row $u$.

The algorithm repeatedly performs the following extension.
Let $i$ be the current endpoint of the open path $P$.

\begin{enumerate}
\item
Query the gradient column
\[
  g=G(\widehat A)_{*,i}.
\]
For every $j\ne i$, compute
\begin{equation}
  c_j=[z^r](\widehat A_{i,j}g_j).
  \label{eq:path-branch-coefficient}
\end{equation}
Choose the first $j$ with $c_j\ne 0$.

\item
Restrict row $i$ to its diagonal entry and the selected arc $i\to j$:
\begin{equation}
  \widehat A'_{i,*}
  =
  e_i^\top+\widehat A_{i,j}e_j^\top,
  \qquad
  \widehat A'_{p,*}=\widehat A_{p,*}\quad(p\ne i).
  \label{eq:path-row-restriction}
\end{equation}
Replace $\widehat A$ by $\widehat A'$, set $c=c_j$, and append $j$
to $P$.

\item
If $j=u$, verify that $P$ is a simple directed cycle of length $r$
in the original graph. If the verification succeeds, return this cycle; otherwise, report failure. If \(j\ne u\), continue from the new endpoint \(j\).
\end{enumerate}

The algorithm performs at most $r-1$ extensions after the initial arc.
If it has not returned a verified cycle within this limit, it reports
failure.

The initial restriction removes the diagonal entry
in row $u$, whereas every subsequent restriction
retains the diagonal entry. This choice serves two
purposes. First, it still forces every selected arc:
once the preceding arc enters a vertex, its column
is occupied, so the corresponding row cannot use
its diagonal entry and must use the selected
outgoing arc.
Second, retaining these diagonal entries keeps
the constant matrix equal to
\[
    \widehat A(0)=I-e_ue_u^\top
\]
throughout the extraction. Thus the same fixed
Boolean border can be used at every step. 
This property allows us to
maintain the required algebraic quantities
efficiently.

\subsection{Correctness}
\label{sec:path-correctness}

We first establish a structural property without assuming
that $r=\ell$.
Fix any weight assignment and candidate length $r$ for
which the extraction is initialized.
We show inductively that, before every extension, the
current open path has distinct vertices.

The initial path consists of an arc with distinct endpoints.
For the induction step, write
\[
    P=(p_0,\ldots,p_k),
    \qquad p_0=u,
    \qquad i=p_k,
\]
and assume that its vertices are distinct.
Row $p_0$ contains only its selected arc to $p_1$.
For $1\le a<k$, row $p_a$ contains its diagonal entry
and its selected arc to $p_{a+1}$.
Every other row agrees with $A$.

Every nonzero permutation monomial must use all selected
arcs. Indeed, row $p_0$ forces the assignment
$p_0\mapsto p_1$. Column $p_1$ is then occupied, so row
$p_1$ cannot use its diagonal entry and must use its selected
arc to $p_2$. Continuing along the path forces every
assignment $p_a\mapsto p_{a+1}$ for $0\le a<k$.
This argument uses only rows preceding the endpoint,
so it remains valid after replacing row $i$ by any vector.

Let $g=G(\widehat A)_{*,i}$.
For $j=p_a$ with $1\le a\le k$, the preceding row
$p_{a-1}$ is forced to use column $j$.
Replacing row $i$ by $e_j^\top$ would force a second,
distinct row to use the same column.
Thus the row-replacement identity gives
\begin{equation}
    g_j
    =
    \Phi\bigl(\widehat A[i\leftarrow e_j^\top]\bigr)
    =
    0
    \qquad
    \bigl(j\in V(P)\setminus\{u\}\bigr).
    \label{eq:path-visited-gradient-zero}
\end{equation}
In particular, $g_i=0$, and every tested branch to a vertex
of $V(P)\setminus\{u,i\}$ has coefficient zero.
A selected nonzero branch therefore either reaches a new
vertex or returns to $u$.
In the latter case the algorithm terminates; in the former,
the row restriction preserves the stated row structure
for the longer simple path. This completes the induction.

Consequently, the endpoint row has not been restricted
before its extension, and row $u$ is never modified after
initialization. Each subsequent restriction preserves
the diagonal entry and changes only off-diagonal entries
divisible by $z$. Hence
\begin{equation}
    \widehat A(0)=I-e_ue_u^\top
    \label{eq:path-constant-matrix}
\end{equation}
throughout every execution of the extraction.

\begin{lemma}
\label{lem:path-extraction}
Suppose that $r = \ell$ is exactly the length of a shortest even directed cycle
and that $\alpha_{e_0}\ne 0$.
Then the algorithm returns a simple directed cycle of
length $\ell$ after exactly $\ell-1$ extensions.
\end{lemma}

\begin{proof}
We maintain the invariant
\begin{equation}
  c=[z^\ell]\Phi(\widehat A)\ne 0.
  \label{eq:path-invariant}
\end{equation}
Initially, the row-replacement identity
\eqref{eq:row-replacement} gives
\begin{equation}
  \Phi(\widehat A)
  =
  z\beta_{e_0}G(A)_{v,u},
  \qquad
  [z^\ell]\Phi(\widehat A)=\alpha_{e_0}\ne 0.
  \label{eq:anchored-coefficient}
\end{equation}

Consider the current open path
\[
    P=(p_0,\ldots,p_k),
    \qquad p_0=u.
\]
By the preceding structural argument, its vertices are
distinct, and every nonzero permutation monomial of
$\widehat A$ must use all selected arcs.

We next identify the contributions of degree $\ell$.
A permutation contributing to $\Phi$ has odd parity and therefore
contains an even cycle.
Every even directed cycle in the original graph has length at least
$\ell$, and every nonfixed vertex contributes one factor of $z$.
Consequently, a contribution of degree $\ell$ consists of one even
directed cycle of length $\ell$ and fixed points elsewhere.
The row restrictions only delete entries, so this remains true for
$\widehat A$.

Let $\mathcal C_\ell$ denote the set of shortest even directed cycles.
The forcing argument therefore gives
\[
  [z^\ell]\Phi(\widehat A)
  =
  \sum_{\substack{C\in\mathcal C_\ell\\E(P)\subseteq E(C)}}
       \prod_{e\in E(C)}\beta_e.
\]
In particular, the invariant implies that some shortest even cycle
contains the entire current path.

We now show that an extension preserving the invariant
always exists.
Let $i=p_k$ and $g=G(\widehat A)_{*,i}$.
Since $i\in V(P)\setminus\{u\}$,
\eqref{eq:path-visited-gradient-zero} gives $g_i=0$.
Multilinearity in row $i$ gives
\[
  \Phi(\widehat A)
  =
  \sum_{j\ne i}\widehat A_{i,j}g_j.
\]
Taking coefficients of degree $\ell$ yields
\begin{equation}
  c=\sum_{j\ne i}c_j.
  \label{eq:path-branch-sum}
\end{equation}
Since $c\ne 0$, at least one branch has $c_j\ne 0$.

For the selected branch, the gradient column $g$ is independent of
row $i$. Hence the restriction
\eqref{eq:path-row-restriction} satisfies
\[
  \Phi(\widehat A')
  =
  g_i+\widehat A_{i,j}g_j
  =
  \widehat A_{i,j}g_j.
\]
It follows that
\[
  [z^\ell]\Phi(\widehat A')=c_j\ne 0,
\]
which proves that the invariant is preserved.


If $j=u$, the selected arcs already form a directed cycle.
Any simple directed cycle containing these arcs must be exactly this
cycle, so its length is $\ell$.
Finally, an open path contained in a simple $\ell$-cycle has at most
$\ell-1$ arcs. Once it has that many arcs, its next extension must
return to $u$.

Thus the algorithm closes a simple directed cycle after selecting
exactly $\ell$ arcs.
\end{proof}

\subsection{Efficient Maintenance}

It remains to implement the gradient-column queries
and row restrictions without recomputing $\Phi$ and
its gradient from scratch.

The structural argument in
Section~\ref{sec:path-correctness} applies to every
execution, whether or not $r=\ell$.
In particular, the endpoint row is unrestricted before
each extension, and the constant matrix is always
given by \eqref{eq:path-constant-matrix}.

The matrix $\widehat A$ is therefore not invertible over
$\mathcal R_r$. We handle this using two fixed bordered matrices.
For an $n\times n$ matrix $X$, define
\begin{equation}
  M_b(X)
  =
  \begin{pmatrix}
    X & e_u\\
    e_u^\top & b
  \end{pmatrix},
  \qquad b\in\{0,1\}.
  \label{eq:anchored-border}
\end{equation}

\begin{lemma}
\label{lem:anchored-border}
For every $X\in\mathcal R_r^{n\times n}$,
\begin{equation}
  \Phi(X)=\Phi(M_0(X))+\Phi(M_1(X)).
  \label{eq:anchored-border-value}
\end{equation}
Moreover, for all $i,j\in[n]$,
\begin{equation}
  G(X)_{j,i}
  =
  G(M_0(X))_{j,i}+G(M_1(X))_{j,i}.
  \label{eq:anchored-border-gradient}
\end{equation}
If $X(0)=I-e_ue_u^\top$, then both bordered matrices are
invertible over $\mathcal R_r$.
\end{lemma}
\begin{proof}
Regard the entries of $X$ as indeterminates, and introduce
a further indeterminate $t$. Set
\[
    M_t(X)=
    \begin{pmatrix}
        X&e_u\\
        e_u^\top&t
    \end{pmatrix}.
\]
The polynomial $\Phi(M_t(X))$ has degree at most one in $t$.
A permutation monomial containing $t$ fixes the added index,
and its restriction to the original indices has the same
parity. Thus the coefficient of $t$ is $\Phi(X)$, and
\[
    \Phi(M_t(X))=H(X)+t\Phi(X)
\]
for some polynomial $H$ independent of $t$. Evaluating at
$t=0,1$ and adding cancels $H$ in characteristic two,
proving the first identity.

This is a formal polynomial identity in the entries of $X$,
with the border held fixed. Since $X_{i,j}$ occurs only in
the corresponding entry of the upper-left block of each
$M_b(X)$, differentiation with respect to $X_{i,j}$ gives
\[
    G(X)_{j,i}
    =
    G(M_0(X))_{j,i}+G(M_1(X))_{j,i}.
\]
Both identities remain valid after substituting arbitrary
entries from $R_r$.

Finally, suppose that $X(0)=I-e_ue_u^\top$.
For either $b\in\{0,1\}$, simultaneously permuting the rows
and columns of $M_b(X)(0)$ gives
\[
    I_{n-1}\oplus
    \begin{pmatrix}
        0&1\\
        1&b
    \end{pmatrix}.
\]
The $2\times2$ block has determinant one in characteristic
two. Hence $\det M_b(X)$ has constant coefficient one and
is a unit in $R_r$, proving that both bordered matrices
are invertible.
\end{proof}

For $b\in\{0,1\}$, maintain
\[
  K_b=M_b(\widehat A),
  \qquad
  B_b=K_b^{-1},
  \qquad
  \Delta_b=\det K_b,
  \qquad
  F_b =\Phi(K_b).
\]
After the initial arc is selected, Lemma~\ref{lem:series} initializes
these quantities in
\[
  \widetilde O(n^\omega(r+1))
\]
field operations.

The following lemma implements both a gradient-column query and the
subsequent row update.

\begin{lemma}
\label{lem:row-update}
Let $A\in\mathcal R_L^{s\times s}$ be invertible, where
$\mathcal R_L=\mathbb F[z]/(z^{L+1})$, and suppose that
$B=A^{-1},\ \Delta=\det A$ and $F=\Phi(A)$
are available.
For an index $i\in[s]$ and a row vector
$\rho\in\mathcal R_L^{1\times s}$, put
\[
  A'=A+e_i\rho,
  \qquad
  g=G(A)_{*,i},
  \qquad
  b=Be_i,
  \qquad
  v=\rho B,
  \qquad
  \eta=1+v_i.
\]
Then
\begin{equation}
  \Phi(A')=\Phi(A)+\rho g,
  \qquad
  \det A'=\Delta\eta.
  \label{eq:row-update-values}
\end{equation}
If $\eta$ is a unit, then
\begin{equation}
  (A')^{-1}=B+\eta^{-1}bv.
  \label{eq:row-update-inverse}
\end{equation}
The column $g$ and all applicable updated quantities can be computed
in $\widetilde O(s^2(L+1))$ field operations.
In particular, $\rho(0)=0$ guarantees that $\eta$ is a unit.
\end{lemma}

\begin{proof}
The row-query identity \eqref{eq:row-query} gives
\[
  g=Bh,
\]
where
\[
  h_i=\Phi(A),
  \qquad
  h_p
  =
  \Delta\sum_{t=1}^s A_{p,t}^2B_{t,p}B_{t,i}
  \quad(p\ne i).
\]
These identities remain valid over $\mathcal R_L$: their derivation
uses multilinearity and identities for invertible matrices over a
commutative ring of characteristic two.
Computing all entries of $h$ and then multiplying by $B$ requires
$O(s^2)$ ring operations.

Multilinearity in row $i$ gives
\[
  \Phi(A+e_i\rho)=\Phi(A)+\rho G(A)_{*,i},
\]
which proves the first identity in \eqref{eq:row-update-values}.
The determinant identity follows from the matrix determinant lemma,
and \eqref{eq:row-update-inverse} is the Sherman--Morrison formula
in characteristic two.

Computing $v$, forming the outer product $bv$, and updating the stored
quantities requires another $O(s^2)$ ring operations and, when needed,
one inversion of a scalar unit.
An operation in $\mathcal R_L$, including inversion of a unit, costs
$\widetilde O(L+1)$ field operations.
This gives the claimed bound.

Finally, if $\rho(0)=0$, then $v(0)=0$, and hence $\eta(0)=1$.
Thus $\eta$ is a unit.
\end{proof}

We apply this lemma to the two bordered matrices at each extension.
For the current endpoint $i$, query
\[
  g_b=G(K_b)_{*,i}
  \qquad(b\in\{0,1\}).
\]
By Lemma~\ref{lem:anchored-border}, the column required by the
extraction algorithm is
\begin{equation}
  g_j=(g_0)_j+(g_1)_j
  \qquad(j\in[n]).
  \label{eq:path-gradient-column}
\end{equation}
The endpoint row has not yet been restricted, so for every original
arc $(i,j)$,
\[
  \widehat A_{i,j}=z\beta_{(i,j)}.
\]
Its branch coefficient can therefore be read as
\[
  c_j=\beta_{(i,j)}[z^{r-1}]g_j.
\]
For a nonarc, the branch coefficient is zero.
Thus, once the two gradient columns are available, all branch
coefficients can be obtained with $O(n)$ additional field operations.

After choosing a branch, let
\[
  \rho
  =
  \bigl(\widehat A'_{i,*}-\widehat A_{i,*},\,0\bigr)
\]
be the corresponding row increment, extended by a zero in the border
coordinate. Then
\[
  K_b'=K_b+e_i\rho
  \qquad(b\in\{0,1\}).
\]
Use Lemma~\ref{lem:row-update} to update $B_b,\Delta_b,F_b$,
using the gradient columns already queried.

The restriction preserves the diagonal entry and only deletes
off-diagonal entries divisible by $z$. Hence $\rho(0)=0$.
Every scalar inverted in the row-update formula consequently has
constant coefficient $1$, so every update is valid.

This implements one extension in
\[
  \widetilde O(n^2(r+1))
\]
field operations. The stored values also satisfy
\[
  \Phi(\widehat A)=F_0+F_1
\]
throughout, by \eqref{eq:anchored-border-value}.

\subsection{Running Time and Success Probability}

\begin{proof}[Proof of Theorem~\ref{thm:witness}]
Suppose first that the graph contains an even directed cycle, and let
$\ell$ be the shortest such length.

By Theorem~\ref{thm:support}, with probability $1-O(n^{-3})$,
the computation in Section~\ref{sec:support} returns $r=\ell$ and
the complete set of arcs belonging to shortest even cycles.
Every returned arc has a nonzero extracted coefficient by the
definition of that algorithm.
On this event, the initial arc can therefore be selected, and
Lemma~\ref{lem:path-extraction} guarantees that the extraction
returns a shortest even directed cycle.

The extraction introduces no additional randomness.
Once the correct length and a nonzero initial coefficient are
available, the identity \eqref{eq:path-branch-sum} guarantees a
nonzero branch at every extension.
The algebraic maintenance is valid for every row restriction made
by the algorithm.
Thus the overall success probability remains $1-O(n^{-3})$.

On this event, finding the length and arc coefficients takes
\[
  \widetilde O(n^\omega\ell)
\]
time.
Initializing the two bordered matrices and their stored quantities
has the same bound.

There are exactly $\ell-1$ extensions, each costing
$\widetilde O(n^2(\ell+1))$.
The total running time is therefore
\[
  \widetilde O(n^\omega\ell)
  +
  \widetilde O\bigl((\ell-1)n^2(\ell+1)\bigr)
  =
  \widetilde O(n^\omega\ell+n^2\ell^2).
\]

In every execution, the candidate length satisfies $r\le n$, and the
algorithm performs at most $r-1$ extensions.
Together with the unconditional bound from
Section~\ref{sec:support}, this gives
\[
  \widetilde O(n^{\omega+1}+n^4)
  =
  \widetilde O(n^4)
\]
time in every execution.
Since the exceptional event has probability $O(n^{-3})$, its
contribution to the expected running time is at most
$\widetilde O(n)$.
Hence the expected running time is also
\[
  \widetilde O(n^\omega\ell+n^2\ell^2).
\]

The final verification guarantees that every returned witness is a
simple even directed cycle in the original graph.
On the success event above, it has the shortest possible length.

Finally, if the graph contains no even directed cycle, the algorithm
of Section~\ref{sec:support} always reports absence.
No extraction is performed, and the total running time is
$\widetilde O(n^{\omega+1})$.
\end{proof}

\section*{Acknowledgments}
We thank Sayan Bhattacharya for polishing the introduction.

\section*{AI methodology}
We used ChatGPT (OpenAI) to assist with literature research,
algorithm design, proof development and checking, and the drafting
and revision of the manuscript. This assistance materially affected
the entire paper.
The authors take full responsibility for
the correctness and originality of all content, including the
mathematical arguments and the accuracy of the references.

\bibliographystyle{alpha}
\bibliography{references}

\end{document}